\documentclass{article}
\usepackage{amsmath,amssymb,amsthm,multirow,graphicx}
\newtheorem{proposition}{Proposition}
\newtheorem{lemma}{Lemma}
\newtheorem{definition}{Definition}

\newtheorem{example}{Example}
\usepackage{smartdiagram}

\begin{document}
\title{Empathy in Bimatrix Games}
\author{  Brian Powers, Michalis Smyrnakis and Hamidou Tembine\thanks{M. Smyrnakis and H. Tembine are with  Learning and Game Theory Laboratory, New York University Abu Dhabi}
\thanks{B. Powers is with Arizona State University, Tempe, Arizona.}
}
\maketitle

\begin{abstract}
Although the definition of what empathetic preferences exactly are is still evolving, there is a general consensus in the psychology, science and engineering  communities
 that the evolution toward players' behaviors in interactive decision-making problems   will be accompanied by the exploitation of their  empathy, sympathy, compassion, antipathy, spitefulness, selfishness, altruism, and self-abnegating states in the payoffs. In this article, we study  one-shot bimatrix games  from a psychological game theory viewpoint. 
A new empathetic payoff model is calculated to fit  empirical observations and  both pure and mixed equilibria are investigated. 
For a realized  empathy  structure, the bimatrix game is categorized among  four generic class of games.
Number of interesting results are derived.  A notable level of involvement can be observed in the empathetic one-shot game  compared the non-empathetic one and this holds 
even for games with dominated strategies. Partial altruism can help in breaking symmetry, in reducing payoff-inequality and in selecting social welfare and more efficient outcomes.
 By contrast, partial spite and self-abnegating may worsen payoff equity. Empathetic evolutionary game dynamics are introduced to capture the resulting empathetic  evolutionarily stable strategies under wide range of revision protocols including Brown-von Neumann-Nash, Smith, imitation, replicator, and hybrid dynamics. Finally, mutual support and Berge solution are investigated and their connection with empathetic preferences are established.  
We show that pure altruism is logically inconsistent, only by balancing it with some partial selfishness does it create a consistent psychology.

 \end{abstract}
\section{Introduction}
We consider two players. Let $\mathcal{N}=\{1,2\}$ be the set of players. Each player $i$ has a set of actions $\mathcal{A}_i.$ The non-empathetic reward functions of player $i$ is $r_i:\ \prod_{i} \mathcal{A}_i \rightarrow \mathbb{R}.$
We consider empathetic preferences. Players  have preferences on the joint strategy outcomes.   
The outcome of player $i$ choosing action $a_i\in \mathcal{A}_i$ combines her intrinsic preference for $a_i\in \mathcal{A}_i$ with the intrinsic preference of $i$'s neighbors, $\mathcal{N}_i$, where the weight given to the preference of any neighbor $j\in \mathcal{N}_i$ depends on the strength of the directed relationship between $i$ and $j.$ A basic setup and for illustration purpose this is modeled  with  the number $\lambda_{ij}$. By a self-regarding player we refer to a player in the game who optimizes her own-payoff (without empathy for others). A self-regarding player 
thus cares about the behavior  that impact her own payoff.  This is scaled with a number $\lambda_{ii}.$ The sign of $\lambda_{ii}$ plays an important role as it determines if it is a maximization or minimization of own-payoff. An other-regarding player considers not only her own payoff  but also some of her network members' payoffs. Then, the other-regarding player will include these in her preferences and create an empathetic payoff. She is still acting to maximize her new empathetic payoff. Based on these basic empathy structures, we construct an empathetic payoff as a combination of payoffs through a matrix $\Lambda=(\lambda_{ij})_{i,j}.$ In contrast to most of existing studies in this field, we include not only positive value of $\lambda_{ij}$  (referred to as partial altruism) but also negative value of $\lambda_{ij}$ (referred to as partial spite). 
The instant empathetic payoff of $i$ is $$ r_{i}^{\Lambda}:= \lambda_{ii}r_{i}+\sum_{j\in \mathcal{N}_i }\lambda_{ij}r_{j},$$ where $r^I_i:=r_i$ denotes the initial non-empathetic payoff of player $i.$  Based on the material payoffs, a player can have empathy/malice for the others and  selfishness/selflessness for herself.

\begin{figure}[htb]
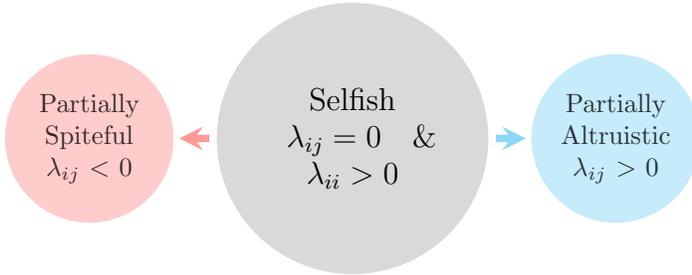

\smartdiagram[constellation diagram]{Selfish ${\lambda_{ij}=0} \quad \& \newline {\lambda_{ii}>0}$, Partially Spiteful $\lambda_{ij}<0$, Partially Altruistic $\lambda_{ij}>0 $ }
 \\
  \caption{Behavior of $i$ towards $j$ for different sign values of $\lambda_{ij}.$ }  \label{fig:diagram}
  \end{figure}

\begin{figure}[htb]
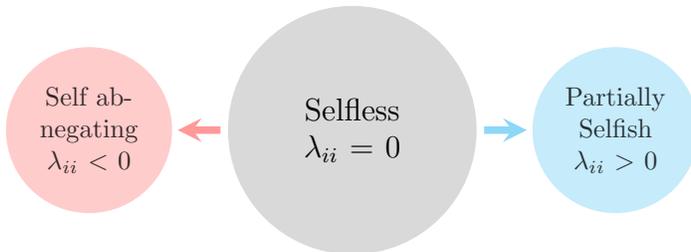

\smartdiagram[constellation diagram]{Selfless $\lambda_{ii}=0$,
  Self abnegating $\lambda_{ii}<0$, Partially Selfish $\lambda_{ii}>0$}\\
  \caption{Behavior of $i$ towards herself for different sign values of $\lambda_{ii}.$ }   \label{fig:diagram2}
  \end{figure}
  The empathetic payoff $r_{i}^{\Lambda}$ can be interpreted as follows.
\begin{itemize} 
\item Selfishness: Selfishness is being concerned  for oneself or one's own payoff, regardless of others' payoffs.  It is the lack of consideration for other players' payoffs.
 If $\lambda_{ij}=0$ we say that $i$ is not empathetic towards $j.$ Player $i$ is self-regarding if $\lambda_{ij}=0$ for all $j\neq i.$ 
If all the $\lambda_{ij}$ are zeros for every $(i,j)$ then, every player focuses on her own-payoff functions positively or negatively depending on the sign of $\lambda_{ii}.$  The case of $\lambda_{ii}>0$  corresponds to the partial selfishness. The case of $\lambda_{ii}<0$ is seen as a self-abnegating behavior (see Figure \ref{fig:diagram2}).

\item Partial Altruism: If $\lambda_{ij}>0 $ we say that $i$ is partially altruistic towards $j.$
If all the $\lambda_{ij}$ are positive  for every $i,j,$ every player is considering the other players in her decision in a partially altruistic way.
\item Partial Spite:  If $\lambda_{ij}< 0$ we say that $i$ is partially spiteful towards $j.$ 
If all the $\lambda_{ij}$ are negative  for every $i,j $ every player is considering the other players in her decision in a partially spiteful way (Figure \ref{fig:diagram}).
\item Mixed altruism-spitefulness-neutrality:  The same player $i$ may have different empathetic behaviors towards her neighbors. For example, if $\lambda_{ij}>0,$  $\lambda_{ik}<0$  and $\lambda_{il}=0$ for $j,k,l\in \mathcal{N}_i$ then player $i$ is partially altruistic towards $j,$  partially spiteful towards $k$ and neutral towards $l.$
\end{itemize}

\subsection{Related work}
In the 1880s,  \cite[pages 102-104]{edge}
 introduced the idea of other-regarding payoff transformations as follows: player $i$  maximizes the payoff function
 $R_i(a_i,a_j)=\lambda r_i+(1-\lambda)r_j$ where $\lambda\in (0,1).$ Here $\lambda$ and $1-\lambda$ represent relative weights that $i$ assigns to $r_i$ (own) and $r_j$(to the other player's) non-empathetic payoff, respectively. 
The work in \cite{ref1} proposed an interesting model of partial altruism as an explanation for the results of public good contribution games, where a player's utility is a linear function of both the player's own monetary payoff and the other players' payoffs. The work in \cite{ref2,mimo} proposed a model that uses both  spite and altruism, where the adjusted utility of a player reflects the player's own utility and his regard for other players. 
A model of fairness is proposed in \cite{ref3} where in addition to purely selfish players, there are players who dislike inequitable outcomes.

\subsection{Contribution}
In this paper we examine  one-shot  $2\times2$ bimatrix games  with empathetic preferences. Our contribution can be summarized as follows. Filling the gap in empathetic games analysis literature, this article  presents empathy from positive, negative or neutral perspective behind the limited focus on altruism or malice. With a clear classification of the game, a better understanding of the overall structure of empathetic outcomes  as well as the learning techniques to evolutionarily stable equilibria are presented.
We show that the altruism enforces  Nash equilibrium payoff equity and improves fairness between the players. In contrast, spite  may worsen the payoff inequality gap. The results reveal that the dominated strategies of the classical bimatrix games are not necessarily dominated any more when users' psychology is involved, and a significant level of involvement can be
 observed among the decision-makers who are positively partially empathetic. Empathy can help in stabilizing to equilibria. Pure altruism is logically inconsistent, only by balancing it with some partial selfishness does it create a consistent psychology.
 \subsection{Structure}
 The rest of the paper is structured as follows. Section \ref{sec:two} presents the empathy structure in games with two players and two actions per player. Section \ref{sec:three} focuses on empathetic evolutionary game dynamics and connection to evolutionarily stable strategies.  Section \ref{sec:four} presents the impact of empathy in generic learning algorithms. Section \ref{sec:3} establishes connection between mutual support, positive empathy altruism and Berge solution concept. Section \ref{incon} examines inconsistency of empathy structure in a multi-level hierarchical reasoning.
Section \ref{sec:five} concludes the paper.
 
\section{Empathy in $2\times2$ Games}\label{sec:two}
We consider two players, each having two actions. Let $\mathcal{A}_i=\{1,2\}$ be action set and $\mathcal{A}:=\prod_{i}\mathcal{A}_i$ be action profiles space of all players. The selfish payoff of player $i$  is denoted by $r_i:\ \mathcal{A} \rightarrow \mathbb{R}.$ By slightly abusing the notation we will write $r_{i}(a)$ for the reward player $i$ gains if the joint action $a \in \mathcal{A}$ is played. Given a $2 \times 2$ matrix $\Lambda$ with entries denoting the empathy of the players and vector $r(a)=\left( \begin{array}{c}
r_{1}(a) \\ r_{2}(a) \end{array} \right)$,  a simplified structure for the empathy reward function, $r^{\Lambda}(a)$, is given by $\Lambda. r(a)$. Therefore, the rewards of the two players will be $( \lambda_{11}r_1+\lambda_{12} r_2$ and $ \lambda_{22}r_2+\lambda_{21} r_1)$ respectively. The empathetic game then, is given by:
$$G_{\Lambda}:=\{  \{1,2\},   \{1,2\}^2, r_1^{\Lambda},r_2^{\Lambda}\}.$$ Table \ref{table:Empathy} represents a generic form of an empathetic game. It is easy to see that if $\Lambda=I$ is the identity matrix, i.e. a diagonal matrix with $\lambda_{ii}=1$, one obtains the no-empathy game with a generic form as it depicted in Table \ref{table:noEmpathy}.
We are interested in the structure of equilibria of the empathetic game $G_{\Lambda}$ for all possible range of the coefficient of the matrix $\Lambda.$

\begin{table*}
\centering
\begin{tabular}{cc|c|c|c|c|l}
\cline{3-4}
& & \multicolumn{2}{ c| }{Player I} \\ \cline{3-4}
& & Left & Right \\ \cline{1-4}
\multicolumn{1}{ |c  }{\multirow{2}{*}{Player II} } &
\multicolumn{1}{ |c| }{Up} & $ (\lambda_{11}a_{11}+\lambda_{12}b_{11},\lambda_{22}b_{11}+\lambda_{21}a_{11})$ & $ (\lambda_{11}a_{12}+\lambda_{12}b_{12},\lambda_{22}b_{12}+\lambda_{21}a_{12})  $      \\ \cline{2-4}
\multicolumn{1}{ |c  }{}                        &
\multicolumn{1}{ |c| }{Down } & $(\lambda_{11}a_{21}+\lambda_{12}b_{21},\lambda_{22}b_{21}+\lambda_{21}a_{21}) $ & $(\lambda_{11}a_{22}+\lambda_{12}b_{22},\lambda_{22}b_{22}+\lambda_{21}a_{22}) $   \\
 \cline{1-4}
\end{tabular}
\caption[]{ $G_{\Lambda}:$Payoff matrix with Empathy}
\label{table:Empathy}
\end{table*}

\begin{table}[htb]
\centering
\begin{tabular}{cc|c|c|c|c|l}
\cline{3-4}
& & \multicolumn{2}{ c| }{Player I} \\ \cline{3-4}
& & Left&  Right  \\ \cline{1-4}
\multicolumn{1}{ |c  }{\multirow{2}{*}{Player II} } &
\multicolumn{1}{ |c| }{Up} & $ (a_{11},b_{11}) $& $ (a_{12},b_{12})  $      \\ \cline{2-4}
\multicolumn{1}{ |c  }{}                        &
\multicolumn{1}{ |c| }{Down } & $(a_{21},b_{21}) $ & $(a_{22},b_{22}) $   \\
 \cline{1-4}
\end{tabular}
\caption[]{ $G_{I}:$ Payoff matrix without empathy.}
\label{table:noEmpathy}
\end{table}
\subsection{Solution Concepts}
 We briefly refer to  few solution concepts. A Nash equilibrium  \cite{nash} is a situation in which no player can improve her payoff by unilateral deviation. A Pareto efficient, or Pareto optimal, is a joint action profile in which it is not possible to make any one player better off without making at least one player worse off.  An evolutionarily stable strategy is a Nash equilibrium which is resilient by small proportion of  deviants (also called mutants). 

\subsection{Classification of Generic $2\times 2$ One-Shot Games }
By suitably choosing the entries of matrix $\Lambda$ the resulting empathetic game of any $2 \times 2$ game will fall in one of the following categories, independently of the rewards' structure of the initial game. These categories include the trivial cases of constant payoff games or games with weakly dominated actions. In addition they can be classified as
coordination (such as Bach-or-Stravinski), anticoordination (such as Hawk-or-Dove), discoordination (such as matching pennies) or games with a dominant strategy (such as Prisoner's dilemma). 
Below we present some existing results for these games in terms of stable or unstable equilibria and limit cycles or oscillations which might occur. In addition vector field plots are used in order to present the evolutionary game dynamics of such games \cite{Tembine2012}.

\subsubsection{Outcomes of coordination games}
\begin{figure}[htb]
\includegraphics[scale=0.3]{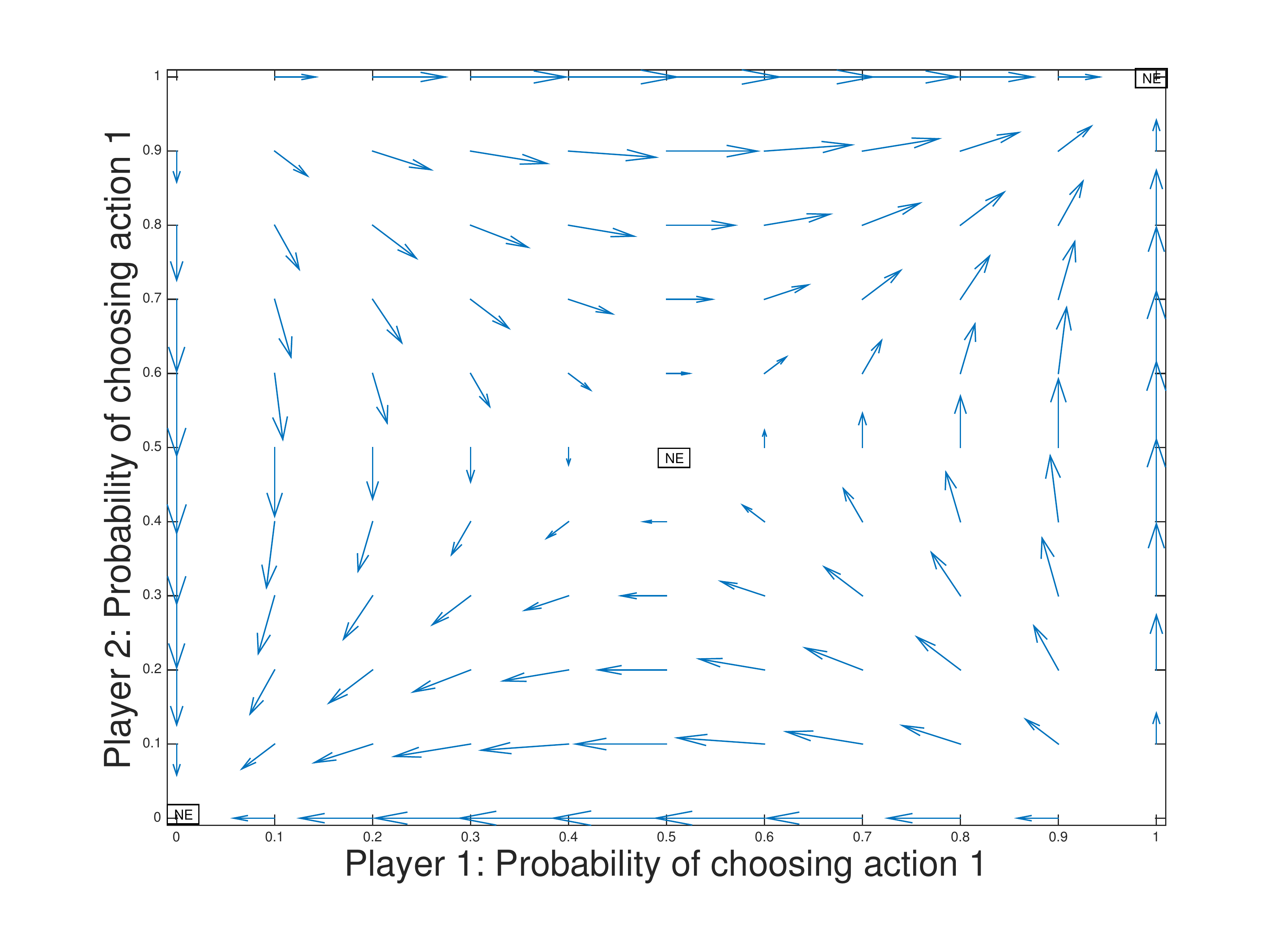}\\ \caption{Outcomes of  coordination games. The two pure Nash equilibria $(0,0)$ and $(1,1)$  are also evolutionarily stable strategies. The mixed Nash equilibrium is not an evolutionarily stable strategy.
 } \label{fig:coord}
\end{figure}
 The canonical example of this class of games is the Bach or Stravinski game or the rendez-vous game. 
 In the generic setting the following inequalities in payoffs hold for Player 1 (rows):  $ a_{11}>a_{21},  a_{22}>a_{12}$ and for Player 2 (columns): 
$ b_{11}>b_{12},  b_{22}>b_{21}.$
 In this game the strategy profiles  $\{Up, Left\}$ and $\{Down, Right\}$ are pure Nash equilibria. This game has two pure strategy Nash equilibria, one where both chose the first action and another where both chose the second action. There is also a mixed strategy Nash equilibrium.
 Figure \ref{fig:coord} depicts a typical vector field of coordination games. As it is illustrated in Figure \ref{fig:coord}, unlike the pure Nash equilibria, the mixed equilibrium is not an evolutionarily stable strategy. Additionally, the mixed Nash equilibrium is also Pareto dominated by the two pure Nash equilibria.  As illustrated in the vector  field, the two pure equilibria are stable and the the fully mixed equilibrium is unstable in the sense of Lyapunov.

\subsubsection{Outcomes of anti-coordination games}
Anti-coordination games have the same properties as coordinations if we change the names of the strategies for Player 2. 
 In this class of games it is mutually  beneficial for the players to play different strategies. An example of a anti-coordination games is the so-called  Hawk and Dove games, or snowdrift game or game of chicken. The payoffs of the players in an anti-coordination have the following properties:
$a_{21}> a_{11}, \ a_{12}> a_{22}$ and $ b_{12}>b_{11},  \ b_{21}>b_{22}$  for row-player and  column-player respectively. The pure action profiles $\{Down, Left\}$ and $\{Up, Right\}$ are the two pure Nash equilibria. There is also a unique mixed strategies Nash equilibrium. Figure \ref{fig:anticoord} plots a typical vector field of an anti-coordination game.   As illustrated in the vector  field, the two pure equilibria are stable and the the fully mixed equilibrium is unstable in the sense of Lyapunov.

\begin{figure}[htb]
\includegraphics[scale=0.3]{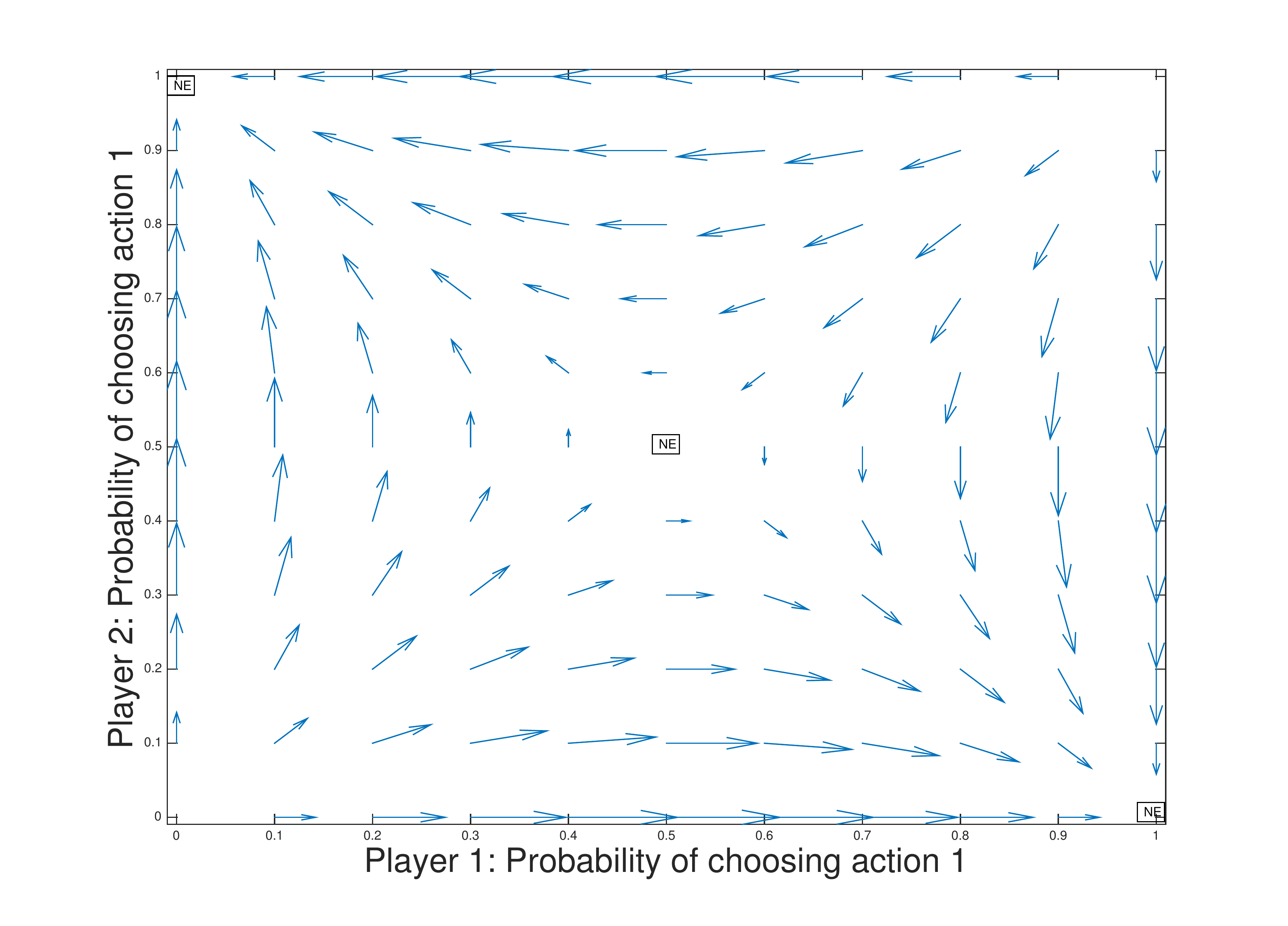}\\ \caption{Outcomes of anti-coordination games. Two pure equilibria and one fully mixed equilibrium} \label{fig:anticoord}
\end{figure}

\subsubsection{Outcomes of empathetic prisoner's dilemma games}
\begin{figure}[htb]
\includegraphics[scale=0.3]{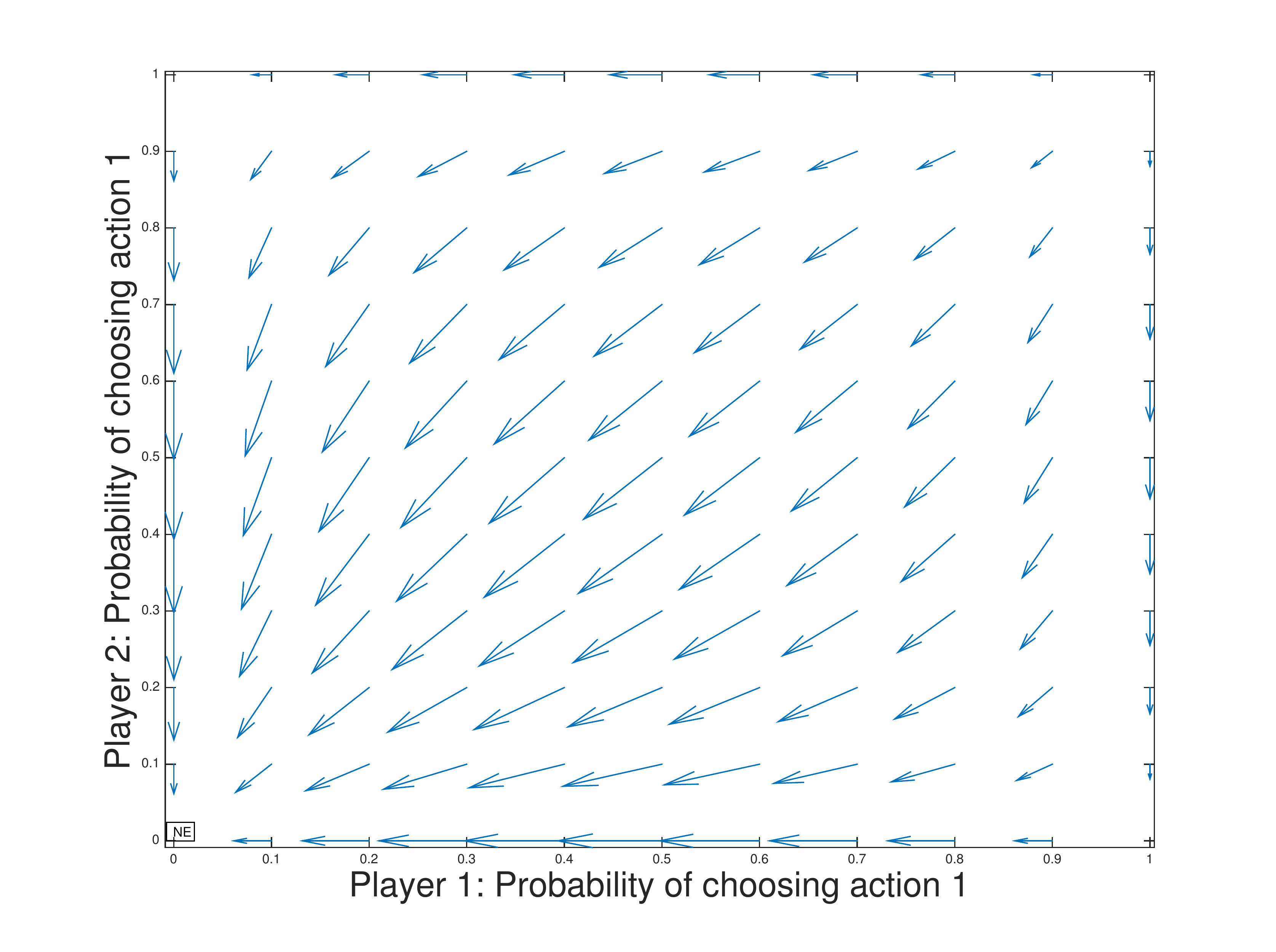}\\ \caption{Outcomes of prisoner's dilemma games. All interior trajectories converges to the unique ESS  $(0,0)$ at the corner.} \label{fig:pd}
\end{figure}

A Prisoner's Dilemma game is a $2 \times 2$ game where each player has a strictly dominant strategy, where the rewards of the players are of the form $a_{21}>a_{11}>a_{22}>a_{12}.$  The outcome where both players choose their dominated strategy strongly Pareto-dominates the outcome where both players choose their dominant strategy. This belongs to the class of games with a dominant strategy. Figure \ref{fig:pd} plots a typical vector field of prisoner's dilemma games. We observe a global convergence to the unique evolutionarily stable strategy under various evolutionary game dynamics.

\subsubsection{Outcomes of empathetic matching pennies games}

In Table \ref{table:noEmpathy}, choosing payoffs so that  $ a_{11}>a_{21}, \ a_{22}>a_{12}$ while $ b_{11}<b_{12}, \ b_{21}>b_{22},$ 
 creates a discoordination game. In each of the four possible action profiles either Player 1 or Player 2 are better off by switching their strategy, so the only Nash equilibrium is a fully mixed one. One such an example of game is the ``so-called''  matching pennies game, which  is played between two players, 1 and 2. Each player has a penny and must secretly turn the penny to heads or tails. The players then reveal their choices simultaneously. If the pennies match (both heads or both tails), then Player 1 keeps both pennies, so wins one from Player 2 (+1 for Player 1, -1 for Player 2). If the pennies do not match (one heads and one tails) Player 2 keeps both pennies, so receives one from Player 1 (-1 for Player 1, +1 for Player 2). There is no pair of pure strategies such that neither player would want to switch if told what the other would do. Instead, the unique Nash equilibrium of this game is in mixed strategies: each player chooses heads or tails with equal probability.  
 Figure \ref{fig:mp} plots a typical vector field of a typical zero-sum matching pennies game. The dynamics   need not converge even if the equilibrium point is the unique equilibrium point of the game. 
 
\begin{figure}[htb]
\includegraphics[scale=0.3]{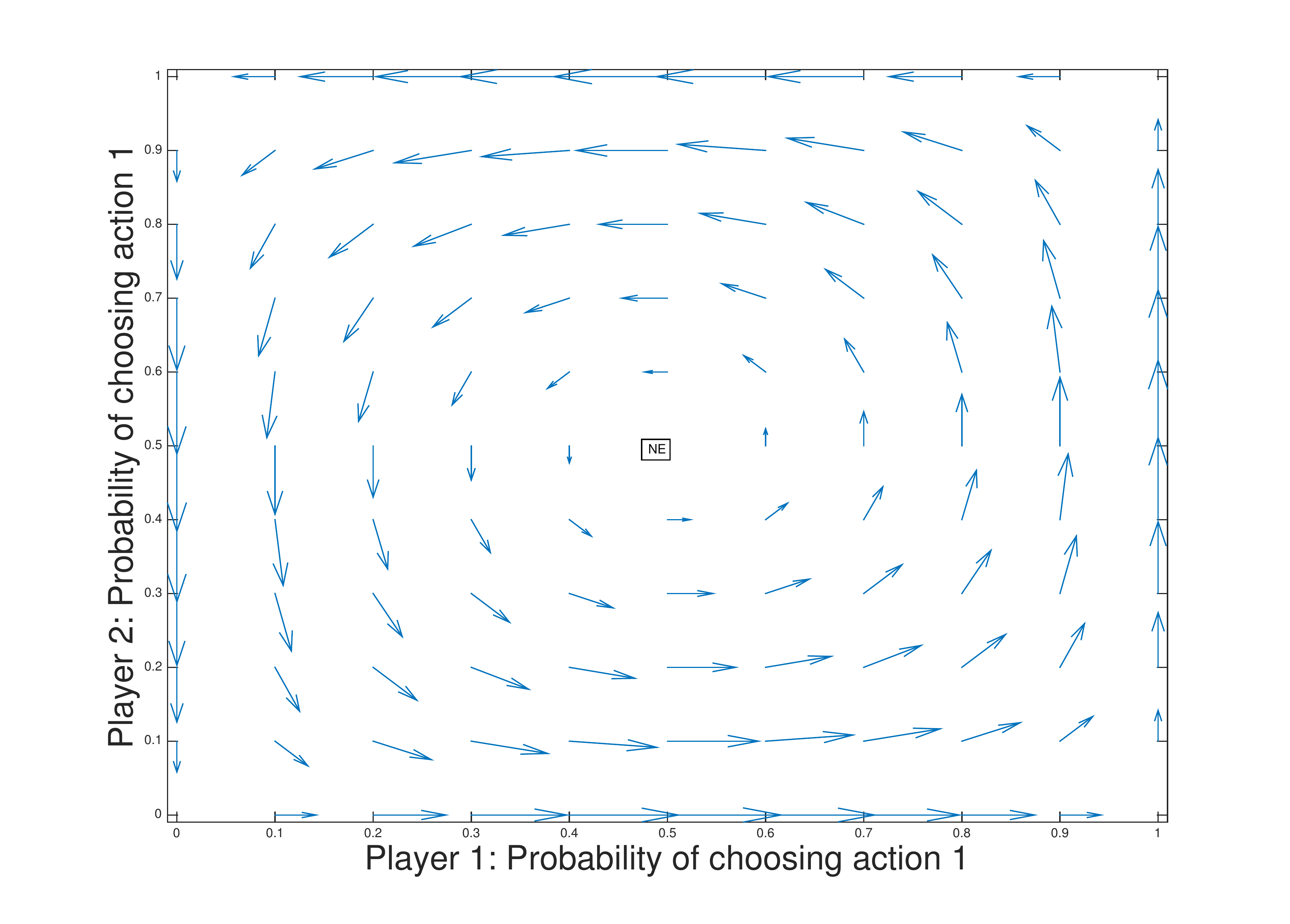}\\ \caption{Outcomes of empathetic  matching pennies games. Non-convergence to the unique mixed strategy equilibrium. Convergence to a limit cycle for $\Lambda=I.$ } \label{fig:mp}
\end{figure}

\section{Effect of empathy on one-shot game's outcome}
In this section various effects of empathy in the games which belong to the categories referred in the previous section are studied. A strategy is dominated for a player if she has another strategy that performs at least as good no matter what other players choose.

The next result shows that dominated strategies of the game without empathy can survive in the empathetic game.

\begin{proposition} A dominated strategy in $G_{I}$ is not necessarily dominated in the game  $G_{\Lambda}$ for $\Lambda\neq I.$
\end{proposition}
This is a very important as it allows the survival of dominated strategies when empathy is involved.
\begin{proof}
It suffices to prove it for the empathetic prisoner's dilemma game. Let $$\Lambda=\left(\begin{array}{ll}  1 & \lambda_{12}\\  \lambda_{21}&1 \end{array}\right),$$ with $\lambda_{ij}\geq 0, i\neq j$  and 
$a_{21}>a_{11}>a_{22}>a_{12}$ and $b_{ij}=a_{ji}.$ The inequalities  $a_{21}>a_{11}$ and $a_{22}>a_{12}$ imply that it is better for Player 1 to choose action 2 independently of the other player.
Similarly, the inequalities  $b_{12}>b_{11}$ and $b_{22}>b_{21}$ imply that it is better for Player 2 to choose action 2 independently of the other player.
Then  action 1 is dominated by action 2 in the non-empathetic game. Now, we check these inequalities in the empathetic game. 
For Player 1:
$a_{11}^{\Lambda}:=\lambda_{11}a_{11}+\lambda_{12}b_{11}=a_{11}(1+\lambda_{12})$ is greater than
$a_{21}^{\Lambda}:=\lambda_{11}a_{21}+\lambda_{12}b_{21}=a_{21}+\lambda_{12}a_{12}$ as long as  $a_{11}-a_{21}+\lambda_{12}(a_{11}-a_{12})>0$ and this is possible because $a_{11}-a_{12}>0$ by assumption. It suffices to consider  $\lambda_{12}>\frac{a_{21}-a_{11}}{a_{11}-a_{12}}>0$ and then action $1$ is not dominated by action 2 of Player 1 in the empathetic game.

For Player 2: 
$b_{11}^{\Lambda}:= \lambda_{22}b_{11}+\lambda_{21}a_{11}=  (1 +\lambda_{21})a_{11},    $  and 
$b_{12}^{\Lambda}:=   \lambda_{22}b_{12}+\lambda_{21}a_{12}=   a_{21}+  \lambda_{21}a_{12}.$ It follows that  $a_{11}-a_{21}+ \lambda_{21}(a_{11}-a_{12})>0$ if $ \lambda_{21}>\frac{a_{21}-a_{11}}{a_{11}-a_{12}}.$
For $\lambda_{12}$ and $\lambda_{21}$ such that $$  \min(\lambda_{21},\lambda_{12}) >\frac{a_{21}-a_{11}}{a_{11}-a_{12}}>0,$$ the first action is not dominated any more in the empathetic game. In the resulting empathetic game the action profile $(1,1)$ is a pure Nash equilibrium that is Pareto-dominant. We conclude that action 1, which is a dominated  strategy in $G_{I}$, is not necessarily dominated in the game  $G_{\Lambda}$ for $\Lambda\neq I.$
\end{proof}

\begin{proposition}  $G_{I}$  is symmetric does not imply that  $G_{\Lambda}$  is symmetric. In particular  $G_{\Lambda}$ helps in breaking symmetry through empathy.
\end{proposition}
\begin{proof} 
The game  $G_{I}$  is said symmetric if the matrix $A$ and $B$ are square matrices and $B$ is the transpose of $A.$ In empathetic games different players may have different empathy structure $i\neq j,\ \lambda_{ij}\neq \lambda_{ji}, $ and the matrix entries $ b_{12}^{\Lambda}=\lambda_{22}b_{12}+\lambda_{21}a_{12}$ may be different than $a_{21}^{\Lambda}:=\lambda_{11}a_{21}+\lambda_{12}b_{21}$ even if $b_{12}=a_{21}, \ b_{21}=a_{12}.$ Hence the empathy structure can break the symmetry in the game.
\end{proof}

\begin{proposition}  Altruism and self-confirming  can help in selecting social welfare.
\end{proposition}
This result is very important in the sense that it allows the possibility for the social welfare to be selected by means of design of the empathy structure. Note, however, that for some other empathy structure the outcome may strictly worsen the social welfare.
\begin{proof} Consider again the symmetric prisoner's dilemma satisfying 
$a_{21}>a_{11}>a_{22}>a_{12}$ and $b_{ij}=a_{ji}.$ The empathetic version of the game selects the social welfare action profile $(1,1)$ for $$  \min(\lambda_{21},\lambda_{12}) >\frac{a_{21}-a_{11}}{a_{11}-a_{12}}>0.$$
\end{proof}

\begin{proposition}Altruism can help in reducing payoff-inequality. Spite can worsen payoff equity.
\end{proposition}

\begin{proof}   The payoff gap is $$r_{1}^{\Lambda}-r_{2}^{\Lambda}=
\lambda_{11}r_1+\lambda_{12} r_2 - \lambda_{22}r_2-\lambda_{21} r_1$$  $$=  (\lambda_{11}-\lambda_{21})r_1+(\lambda_{12}- \lambda_{22})r_2.$$
For $\Lambda=\left(\begin{array}{ll}  1 & \mu\\  \mu&1 \end{array}\right),$  the payoff gap yields $r_{1}^{\Lambda}-r_{2}^{\Lambda}=(1-\mu) (r_1-r_2).$ This proves both announced results depending on the magnitude of $\mu.$ 
\end{proof}

Let $\Lambda=\left(\begin{array}{ll}  1 & \lambda_{12}\\  \lambda_{21}&1 \end{array}\right)$, $\tilde{\lambda}=\frac{\lambda_{12}}{\lambda_{21}}$ and $a,b$ denote the rewards of Player 1 and 2 respectively for a specific joint action. Then the following proposition holds:
 
\begin{proposition}
\begin{itemize}
\item If $a>b$ and $b>0$ then payoff inequality is reduced if $\tilde{\lambda}<\frac{a}{b}$ and payoff inequality is increased if $\tilde{\lambda}>\frac{a}{b}$ 
\item If $a>b$ and $b<0$ then payoff inequality is reduced if $\tilde{\lambda}>\frac{a}{b}$ and payoff inequality is increased if $\tilde{\lambda}<\frac{a}{b}$  
\item If $a<b$ and $b>0$ then payoff inequality is reduced if $\tilde{\lambda}<\frac{b}{a}$ and payoff inequality is increased if $\tilde{\lambda}>\frac{b}{a}$ 
\item If $a>b$ and $b<0$ then payoff inequality is reduced if $\tilde{\lambda}>\frac{b}{a}$ and payoff inequality is increased if $\tilde{\lambda}<\frac{b}{a}$  
\end{itemize}
\end{proposition}

\begin{proof}   
We will show the proof for the first statement since the rest can be derived using an identical process. The difference between the rewards of the empathetic game and the strategic form game are:  
$$a+\lambda_{12}b-b-\lambda_{21}a<a-b\Leftrightarrow \lambda_{12}b<\lambda_{21}a\Leftrightarrow \tilde{\lambda}<\frac{a}{b}$$. 
Which concludes the proof. 
\end{proof}

\subsection{Empathetic Evolutionarily Stable Strategies} \label{sec:three}
Consider a large population of players. Each player has finite set of actions. Denote by $r_a(m):=r(a,m)$ the payoff function of a generic player with action $a$ when facing a population distribution of actions as $m.$ In the context of pairwise interaction, this is re-interpreted as acting with another player with strategy $m.$ When Player 1 plays a mixed strategy $x,$
the expected payoff of Player 1 is $\langle x, r(m)\rangle$ where $r(m)=(r_a(m))_{a}.$ In the presence of empathy, the payoff is $r^{\Lambda}_a(m).$

 \subsubsection{Homogeneous population}

We consider an homogeneous population of players. That is, $\lambda_{ii}=\sigma$ for all $i$  and $\lambda_{ij}=\mu$ for all $i\neq j.$  The empathy structure is $ \left(\begin{array}{ll}  \sigma & \mu\\  \mu &\sigma \end{array}\right).$ At each time step, two players are randomly selected for a   $2\times 2$ game.  The empathetic payoffs of Player 1 is given by 
$$A^{\Lambda}=\left(\begin{array}{cc} (\lambda_{11}+\lambda_{12})a_{11}& \lambda_{11}a_{12}+\lambda_{12}a_{21} \\ 
\lambda_{11}a_{21}+\lambda_{12}a_{12} & 
(\lambda_{11}+\lambda_{12})a_{22}  \end{array}\right)$$
$$
=
\left(\begin{array}{cc} (\sigma+\mu)a_{11}& \sigma a_{12}+\mu a_{21} \\ 
\sigma a_{21}+\mu a_{12} & (\sigma+\mu) a_{22}  \end{array}\right),
$$
$$
=:\left(\begin{array}{cc} a_{11}^{\Lambda}&  a_{12}^{\Lambda} \\ 
 a_{21}^{\Lambda} & a_{22}^{\Lambda}  \end{array}\right),
$$ and the payoff of Player 2 is  the transpose of the payoff of Player 1.There is a constraint for each player $C m\leq V$ where $C=(c_1, c_2), \ V\in\mathbb{R}, (m_1,m_2)\in\Delta_1,$ $\Delta_1=\{(y,1-y)\ |\ 0\leq y\leq 1 \}.$  Denote by $$\mathcal{C}:=\{ y \ |\ C\left(\begin{array}{cc} y \\ 1-y  \end{array}\right)\leq V\}.$$ The set of constrained best responses to any opponent strategy $m\in\mathcal{C}$ is $$CBR(m)=\arg\max_{x\in\mathcal{C}}\ (x,1-x)A \left(\begin{array}{cc} m \\ 1-m  \end{array}\right).$$
If $c_1=c_2$ the constraint is independent of the strategies. Therefore, we assume that $c_1\neq c_2.$
We assume that $\mathcal{C}$ is non-empty. Denote $\alpha=\frac{V-c_2}{c_1-c_2}.$ 
Denote by $\beta_1=(\lambda_{11}+\lambda_{12})a_{11}-(\lambda_{11}a_{21}+\lambda_{12}a_{12}),\beta_2= (\lambda_{11}+\lambda_{12})a_{22}- (\lambda_{11}a_{12}+\lambda_{12}a_{21}).$
We transform the matrix $A^{\Lambda}$ to the following matrix: $$\bar{A}^{\Lambda}=\left(\begin{array}{cc}\beta_1& 0 \\ 0 & 
\beta_2  \end{array}\right).$$ The next result shows that the two matrix games $\bar{A}^{\Lambda}$ and $ A^{\Lambda}$ have the same Nash equilibrium properties. 

\begin{proposition}
The two matrix games  $$ A^{\Lambda}=\left(\begin{array}{cc} a_{11}^{\Lambda}& a_{12}^{\Lambda}\\  a_{21}^{\Lambda}& a_{22}^{\Lambda}
  \end{array}\right)\ 
\mbox{
and}\  \bar{A}^{\Lambda}=\left(\begin{array}{cc} a_{11}^{\Lambda}-a_{21}^{\Lambda}& 0\\  0& a_{22}^{\Lambda}-a_{12}^{\Lambda},
  \end{array}\right),
$$ 
have the same Nash equilibrium properties in symmetric strategies.
\end{proposition}
\begin{proof}
We compute the set of the possible Nash equilibria of $A^{\Lambda}$ in symmetric strategies. 
Let the population profile be $(x,1-x)$ with $0\leq x\leq 1.$ The pure strategy $x=1$ is an equilibrium if $a_{11}^{\Lambda}\geq a_{21}^{\Lambda}$ i.e.,  $\beta_1:=a_{11}^{\Lambda}-a_{21}^{\Lambda}\geq 0.$
The pure strategy $x=0$ is an equilibrium if $a_{22}^{\Lambda}\geq a_{12}^{\Lambda}$ i.e.,  $\beta_2:=a_{22}^{\Lambda}-a_{12}^{\Lambda}\geq 0.$
An interior equilibrium (whenever it exists) is obtained if the indifference condition is fulfilled. 
\begin{eqnarray}
a_{11}^{\Lambda}x+ a_{12}^{\Lambda}(1-x)= a_{21}^{\Lambda}x+ a_{22}^{\Lambda} (1-x),\\
(a_{11}^{\Lambda}-a_{21}^{\Lambda})x=  (a_{22}^{\Lambda}-a_{12}^{\Lambda}) (1-x),\\
\beta_1 x=\beta_2(1-x)\\
x=\frac{\beta_2}{\beta_1+\beta_2}.
\end{eqnarray}
It turns out that the symmetric equilibria are all obtained by comparing $\beta_1\geq 0$, $\beta_2\geq 0$ or $(\frac{\beta_2}{\beta_1+\beta_2},\frac{\beta_1}{\beta_1+\beta_2}).$ Thus, it has the same equilibrium structure as in  the diagonal matrix  $ \bar{A}^{\Lambda}=\left(\begin{array}{cc} \beta_1& 0\\  0& \beta_2
  \end{array}\right).
$  This completes the proof.
\end{proof}
Notice that the matrix $A$ may not be symmetric.
 Note that if both $\beta_1$ and $\beta_2$ are zero, the transformed payoffs are constant (degenerate case) and hence, any strategy in $\mathcal{C}$ is an equilibrium. None of these equilibria is resilient by small proportion of deviants. Thus, there is no constrained ESS in this case. 
 
 \begin{proposition}
 Any generic empathetic $2\times 2$ matrix game (with non-trivial  payoffs) has at least one constrained ESS. 
 \end{proposition}
 
 Below we prove this statement. By generically, we mean  that $\beta_1\beta_2\neq 0.$  We distinguish two cases depending on the coefficient $c_1$ and $c_2.$
\begin{itemize} \item Let $c_1>c_2.$ If $V>c_1,$ then $\mathcal{C}=[0,1]$ (unconstraint case, there is an ESS). If $V\leq c_1,$ then $\mathcal{C}=[0,\alpha]\subset [0,1].$ \item Let $c_1<c_2$. If $V<c_1,$ $\mathcal{C}=\emptyset$ (excluded by hypothesis). If $V\geq c_1,$ $\mathcal{C}=[\alpha, 1].$ 
    \end{itemize}
    We then have to examine two types of constraints: 
     $$\mbox{Type I:}\ c_1>c_2,\ V<c_1,\mathcal{C}=[0,\alpha],$$  $$\mbox{Type II:}\ c_1< c_2,\ V>c_1,\mathcal{C}=[\alpha, 1]$$
    \subsubsection{Type I}
    \begin{itemize} \item Consider the following setup:
     $$\bar{A}^{\Lambda}=\left(\begin{array}{cc} \beta_1 & 0 \\ 0 & \beta_2  \end{array}\right), \beta_1>0,\beta_2\leq 0, \mathcal{C}=[0,\alpha].$$ The first strategy dominates the second one in the unconstrained game. Hence, the  mixed strategy $m=\alpha$ is the unique ESS in the constrained game. Note that $m=\alpha$ is not an ESS in the unconstrained game.     \item $$\bar{A}^{\Lambda}=\left(\begin{array}{cc} \beta_1 & 0 \\ 0 & \beta_2  \end{array}\right), \beta_1\leq 0,\beta_2> 0, \mathcal{C}=[0,\alpha].$$ The second strategy dominates the first one in the unconstrained game. Hence, the strategy mixed strategy $m=0$ is the unique ESS in the constrained game. This situation belongs to the  class of  prisoner's Dilemma games.
    \item We now swap the sign of $\beta$  with  $$\bar{A}^{\Lambda}=\left(\begin{array}{cc} \beta_1 & 0 \\ 0 & \beta_2  \end{array}\right), \beta_1>0,\beta_2> 0, \mathcal{C}=[0,\alpha].$$  This situation belongs to the  the class of  Coordination Games. The first pure action is being eliminated by the constraint, the second strategy is the unique ESS.
        \item We now look when both $\beta_1$ and $\beta_2$ are negative: $$\bar{A}^{\Lambda}=\left(\begin{array}{cc} \beta_1 & 0 \\ 0 & \beta_2  \end{array}\right), \beta_1<0,\beta_2< 0, \mathcal{C}=[0,\alpha].$$ This situation belongs to the class of Hawk-Dove games.
If $\frac{\beta_2}{\beta_1+\beta_2}\geq \alpha$ then $(\frac{\beta_2}{\beta_1+\beta_2},\frac{\beta_1}{\beta_1+\beta_2})$ is an ESS; else if $\frac{\beta_2}{\beta_1+\beta_2}>\alpha$ then the constrained best response set is
         $$CBR(m)=\arg\max_{x\in\mathcal{C}}\ (x,1-x)A^{\Lambda} \left(\begin{array}{cc} m \\ 1-m \end{array}\right)$$ $$=
         \left\{\begin{array}{cc} \alpha & \mbox{if}\ m< \frac{\beta_2}{\beta_1+\beta_2}, \\
         0 & \mbox{if}\ m> \frac{\beta_2}{\beta_1+\beta_2}, \\
         \mathcal{C}  & \mbox{if}\ m=\frac{\beta_2}{\beta_1+\beta_2}.
          \end{array}\right.$$
         Thus, $CBR(\alpha)=\{\alpha\}$ and $\alpha$ is a constrained ESS.
 \end{itemize}

    \subsubsection{Type II}
    If $\alpha>1,$ the interval $[\alpha,1]$  is empty. We now suppose that $V>c_1$ and $\alpha <1.$
    \begin{itemize}
    \item $\bar{A}^{\Lambda}=\left(\begin{array}{cc} \beta_1 & 0 \\ 0 & \beta_2  \end{array}\right), \beta_1>0,\beta_2\leq 0, \mathcal{C}=[\alpha,1].$ The first strategy $m=1$ is the unique ESS in the constrained game.
    \item $\bar{A}^{\Lambda}=\left(\begin{array}{cc} \beta_1 & 0 \\ 0 & \beta_2  \end{array}\right), \beta_1\leq 0,\beta_2> 0, \mathcal{C}=[\alpha,1]$  the strategy mixed strategy $m=\alpha$ is the unique constrained ESS.
    \item $\bar{A}^{\Lambda}=\left(\begin{array}{cc} \beta_1 & 0 \\ 0 & \beta_2  \end{array}\right), \beta_1>0,\beta_2> 0, \mathcal{C}=[\alpha,1]$. The second strategy $m=0$ does not satisfy the constraint. The first strategy  is a constrained ESS.
        \item $\bar{A}^{\Lambda}=\left(\begin{array}{cc} \beta_1 & 0 \\ 0 & \beta_2  \end{array}\right), \beta_1<0,\beta_2< 0, \mathcal{C}=[0,\alpha].$
  The mixed strategy $\min\left(\frac{\beta_2}{\beta_1+\beta_2},\alpha\right)$ is a constrained ESS.
\end{itemize}

 \subsubsection{Heterogeneous population}
 Consider a  population game characterized by a payoff function: $$r_i(., .): \ \mathcal{A}_i\times \prod_j \mathcal{P}(\mathcal{A}_j)\rightarrow \mathbb{R},$$ where $\mathcal{A}_i$ is finite (and non-empty) and  $\mathcal{P}(\mathcal{A}_i)$ is the space of probability measures over $\mathcal{A}_i.$ The probability vector $m_i\in \mathcal{P}(\mathcal{A}_i)$ represents the aggregative population state of $i$, i.e., the fraction of players per action at population $i$.  We denote the payoff function a generic player of subpopulation $i$ as $ r_i (a, m)=: r_{ia}(m).$ Collecting together one obtains a vector payoff function $r(m)= (r_{ia}(m))_{i, a_i\in \mathcal{A}_i}.$ The empathetic payoff function vector is  $r^{\Lambda}.$
A Nash equilibrium  of the empathetic game  is a population profile $m^*$  that satisfies the following variational inequality:  for every $i,$ $$\langle m^*-m, r^{\Lambda}(m^*)\rangle \geq 0,\ \forall\ m_i\in\mathcal{P}(\mathcal{A}_i).$$
Assuming  that for every action $ a_i\in \mathcal{A}_i,$  the function $ m \longmapsto r^{\Lambda}_i (a, m)=: r^{\Lambda}_{ia}(m),$ is continuous, one can easily show that the  population game has at least one Nash equilibrium. The proof uses a direct application of Brouwer fixed-point theorem and is therefore omitted.

\subsection{Empathetic Learning}  \label{sec:four}
 We   introduce a way of revising the actions of a player from subpopulation $p$ called "revision protocol" as $\eta^p_{ab}(m, r^{\Lambda}(m))\geq 0$ which represents the {\it rate of switching} from action $a$ to $b$ when the entire population profile is $m.$ A population profile together with a learning rule (revision protocol) and a learning rate sequence $\hat{\lambda}$ defines a discrete-time game dynamics $L_{\eta}$, given by,

%\begin{equation}\label{SDEutt2}
%\left\{\begin{array}{lll}
%m_{a,t+1}&=&m_{a,t} \\  & &+\ \lambda_t\sum_{b\in \mathcal{A}}m_{b,t} \eta_{ba}(m_t, r(m_t))\\   & &- \lambda_t m_{a,t}\sum_{b\in \mathcal{A}} \eta_{ab}(m_t, r(m_t)),\\
% & & m_{a,0}\geq 0,\  \sum_{a\in \mathcal{A}} m_{a,0} =1.
%\end{array}\right.
%\end{equation}

\begin{equation}\label{SDEutt6}
\left\{\begin{array}{lll}
m_{a,t+1}^p&=&m_{a,t}^p \\  & &+\ \hat{\lambda}_t\sum_{b\in \mathcal{A}^p}m_{b,t}^p \eta_{ba}^p(m_t, r^{\Lambda}(m_t))\\   & &- \hat{\lambda}_t m^p_{a,t}\sum_{b\in \mathcal{A}^p} \eta_{ab}^p(m_t, r^{\Lambda}(m_t)),\\
 & & m_{a,0}\geq 0,\  \sum_{a\in \mathcal{A}^p} m_{a,0} =1,
\end{array}\right.
\end{equation}
where $\hat{\lambda}_t\geq 0$  is the learning rate sequence. Note that the learning dynamics is well-defined for arbitrary learning rate $\hat{\lambda}_t\geq 0.$

In view of (\ref{SDEutt6}), the first term describes the population state at the previous time-step,
the second term represents the  inflow into the action $a$ from other actions, whereas
the third term provides the outflow from action $ a$ to other actions. The difference between  
these last two terms is the  change in the use of action $a$, that added
to the original proportion provides us with the new proportion of use of action $a.$
We now check that  (\ref{SDEutt6}) is well-defined over the simplex $\mathcal{P}(\mathcal{A})$. 
The lemma below states that if the starting point is inside the domain  $\mathcal{P}(\mathcal{A})$ then the dynamics will remain inside $\mathcal{P}(\mathcal{A}):$ the dynamics is forward invariant.
\begin{lemma} For a well-designed learning rates  $\hat{\lambda}_t,$ 
The simplex  $\mathcal{P}(\mathcal{A})$ is forward invariant under (\ref{SDEutt6}), i.e., if initially $m_0\in \mathcal{P}(\mathcal{A})$ then for every $t\geq 0$, the solution of (\ref{SDEutt6}),  $m_t\in \mathcal{P}(\mathcal{A}).$
\end{lemma}

%\begin{proof} For two actions per player, the heterogeneous game dynamics yields
%\begin{equation}\label{SDEutt6t}
%\left\{\begin{array}{lll}
%m_{1,t+1}^1&=&m_{1,t}^1 \\  & &+\ \hat{\lambda}_t (1-m^1_{1,t}) \eta_{21}^1\\   & &- \hat{\lambda}_t m^1_{1,t}\eta_{12}^1,\\
%m_{1,t+1}^2&=&m_{1,t}^2 \\  & &+\ \hat{\lambda}_t (1-m^2_{1,t}) \eta_{21}^2\\   & &- \hat{\lambda}_t m^2_{1,t}\eta_{12}^2,\\
% & & m^1_{1,0}\in [0,1],\  m^2_{1,0}\in [0,1].
%\end{array}\right.
%\end{equation}
%\end{proof}

\begin{proposition}
If the empathetic game belongs to the class of coordination, anticoordination or prisoner's dilemma, there is a convergence to the set of pure equilibria under Brown-von Neumann-Nash (BNN), Smith and replicator dynamics.
\end{proposition}
The proof is immediate as illustrated.   For two actions per player, the heterogeneous game dynamics yields
\begin{equation}\label{SDEutt6t}
\left\{\begin{array}{lll}
m_{1,t+1}^1&=&m_{1,t}^1 \\  & &+\ \hat{\lambda}_t (1-m^1_{1,t}) \eta_{21}^1\\   & &- \hat{\lambda}_t m^1_{1,t}\eta_{12}^1,\\
m_{1,t+1}^2&=&m_{1,t}^2 \\  & &+\ \hat{\lambda}_t (1-m^2_{1,t}) \eta_{21}^2\\   & &- \hat{\lambda}_t m^2_{1,t}\eta_{12}^2,\\
 & & m^1_{1,0}\in [0,1],\  m^2_{1,0}\in [0,1].
\end{array}\right.
\end{equation}

The behavior under these dynamics are illustrated in Figures \ref{fig:coord}, \ref{fig:anticoord}, \ref{fig:pd}. In addition, the homogeneous population dynamics converges to the unique ESS in anticoordination games. It remains to analyze the class of matching pennies games.
In the matching pennies  we find that even if the game has a unique mixed strategy equilibrium, the equilibrium point may be unstable in the sense that for any initial condition (other than that equilibrium point), the system never converges to the equilibrium point. One innovative result is that, thanks to empathy,  this instability can be broken when both players are empathetic  with opposite signs.

\begin{proposition}
Under specific empathy matrix $\Lambda$ matching pennies game can be tranformed to a coordination game. Therefore, empathy can be used in order to breaking the instability of the matching pennies game.
\end{proposition}

\begin{proof}
Let the non-empathetic matching pennies game be $ a_{11}>a_{21}, \ a_{22}>a_{12}$ while $ b_{11}<b_{12}, \ b_{21}>b_{22}.$ 
Let $\lambda_{11}>0, \lambda_{12}<0$ while  $\lambda_{22}<0, \lambda_{21}>0.$ Then, 
Then the following inequalities hold: 
$$a_{11}^{\lambda}:=\lambda_{11}a_{11}+\lambda_{12}b_{11} > a_{21}^{\lambda}:=\lambda_{11}a_{21}+\lambda_{12}b_{21}.$$
and 
$b_{11}^{\lambda}:= \lambda_{22}b_{11}+\lambda_{21}a_{11}$ is greater than
$b_{12}^{\lambda}:=   \lambda_{22}b_{12}+\lambda_{21}a_{12}.$ It is turns out the pure $(1,1)$ becomes a strict Nash equilibrium in the empathetic game 
with  $\lambda_{11}>0, \lambda_{12}<0, \lambda_{22}<0, \lambda_{21}>0.$ This stabilizes the empathetic evolutionary dynamics to a new equilibrium.
\end{proof}

The payoff matrix of matching pennies game with the following empathy  structure $\lambda_{11}>0, \lambda_{12}<0, \lambda_{22}<0, \lambda_{21}>0$ leads to a coordination game (see Table \ref{table:Empathy2}). Hence this leads to  two stable pure equilibria at the corner $(1,1)$ and $(0,0)$ as displayed in Figure \ref{fig:coord}.

\begin{table}
\centering
\begin{tabular}{cc|c|c|c|c|l}
\cline{3-4}
& & \multicolumn{2}{ c| }{ I} \\ \cline{3-4}
& & L & R \\ \cline{1-4}
\multicolumn{1}{ |c  }{\multirow{2}{*}{ II} } &
\multicolumn{1}{ |c| }{U} & $ (\lambda_{11}-\lambda_{12},-\lambda_{22}+\lambda_{21})$ & $ (-\lambda_{11}+\lambda_{12},\lambda_{22}-\lambda_{21})  $      \\ \cline{2-4}
\multicolumn{1}{ |c  }{}                        &
\multicolumn{1}{ |c| }{D } & $(-\lambda_{11}+\lambda_{12},\lambda_{22}-\lambda_{21}) $ & $(\lambda_{11}-\lambda_{12},-\lambda_{22}+\lambda_{21}) $   \\
 \cline{1-4}
\end{tabular}
\caption[]{ The payoff bimatrix of matching pennies game with empathy $\lambda_{11}>0, \lambda_{12}<0, \lambda_{22}<0, \lambda_{21}>0$ leads to anticoordination game. Hence two stable pure equilibria }
\label{table:Empathy2}
\end{table}

Figure \ref{fig:mp1} illustrates  cycling behavior  in two-population matching pennies game under replicator dynamics with starting point $(0.4,0.6),(0.6,0.4)$ and $\Lambda=I.$ Figure \ref{fig:emp2} illustrates an elimination of limit cycle in an empathetic two-population matching pennies game with empathy structure  under replicator dynamics with starting point $(0.55,.45),(0.65,0.35)$ by changing the empathy structure to be  $\Lambda=\left(\begin{array}{cc}1 & 0.0001 \\ 0.0001 & 
-1\end{array}\right)$ 

%% {\color{red} please refer here to the actual numbers of $\Lambda$ used to produce these figures, which algorithm was used in $\eta$  of the learning algorithm in (6) and possible initial conditions.}

\begin{figure}
\includegraphics[scale=0.5]{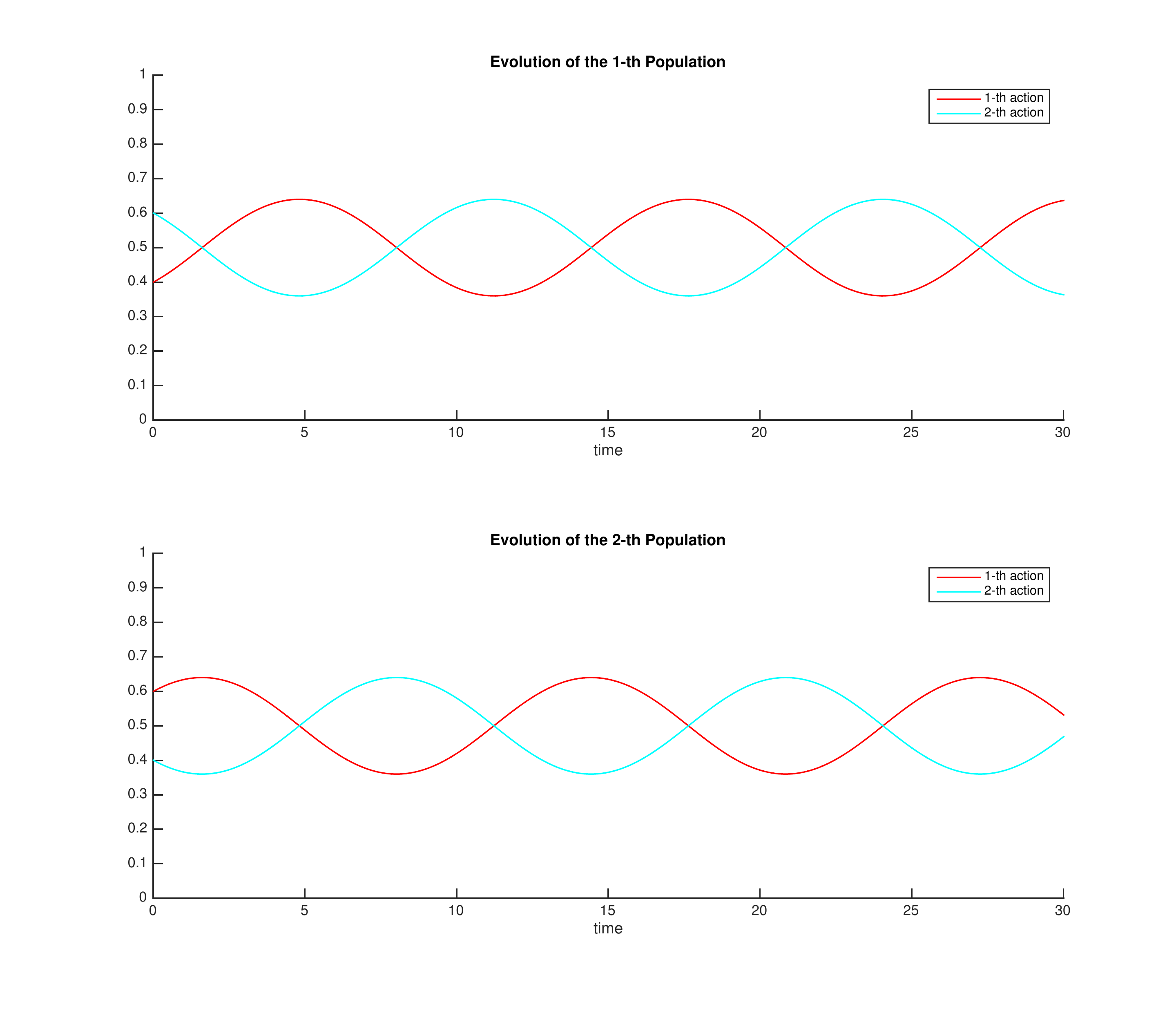}\\ \caption{Cycling in two-population matching pennies game under replicator dynamics.
 } \label{fig:mp1}
\end{figure}

\begin{figure}
\includegraphics[scale=0.5]{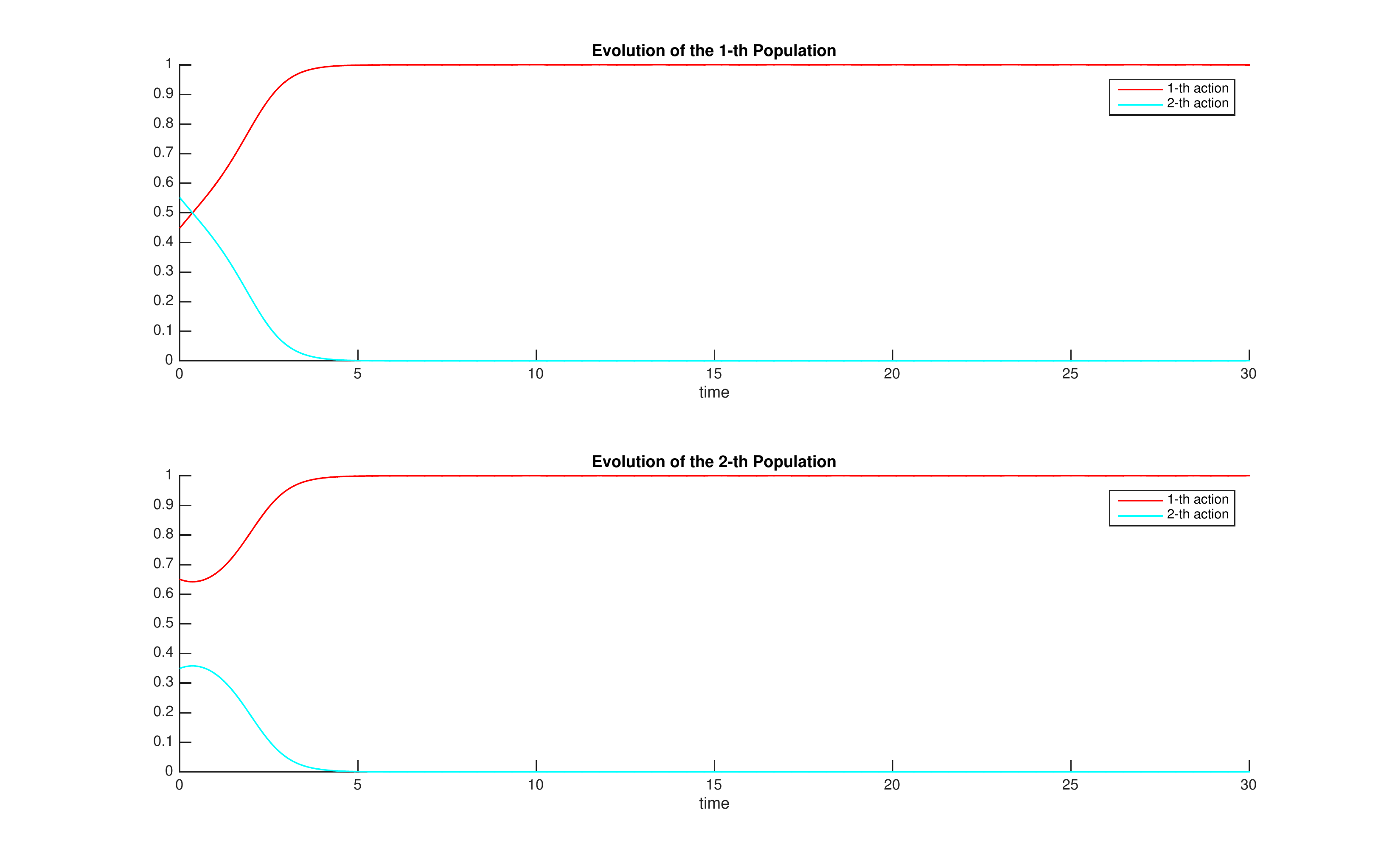}\\ \caption{Elimination of limit cycle in  an empathetic two-population matching pennies game under replicator dynamics.
 } \label{fig:emp2}
\end{figure}

\begin{proposition} In games with a dominant strategy, the involvement of positive empathy may permit the survival of the dominant strategy. Moreover, the presence of the non-neutral empathy may change the equilibrium structure.
\end{proposition}
As a corollary, cooperative behaviors can be observed even in one-shot prisoner's dilemma games (and hence breaking the dilemma). As illustrated in the diagram of  Figure \ref{fig:pd1}, all kind of equilibria are possible in the 
empathetic game depending on the values of $\lambda_{12}$ and $\lambda_{21}.$

\begin{figure}[htb]
\includegraphics[scale=0.4]{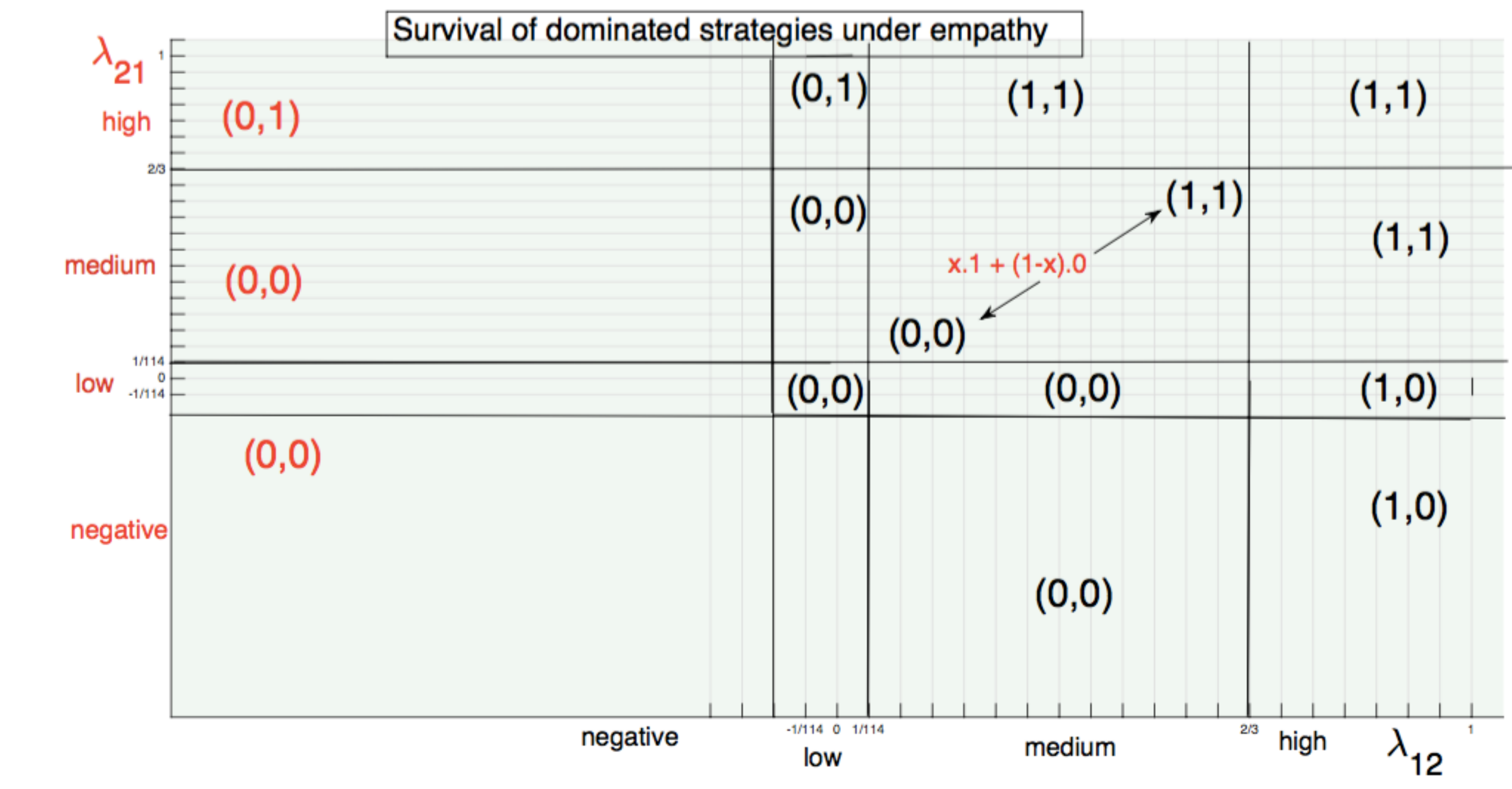}\\ \caption{Different  equilibrium outcomes in the empathetic game. Survival of the (initially) dominated strategy
 } \label{fig:pd1}
\end{figure}

\section{Mutual support and Berge solution} \label{sec:3}
The Berge solution concept was introduced in \cite[page 20]{berge}. See also \cite{reft16,reft17,reft18,reft19,reft20,reft21,reft22} for recent investigation of Berge solution.
The strategy profile $a^*$ is a Berge solution if
$$
r_i(a^*)=\max_{a_{-i}}r_i(a_i^*,a_{-i}).
$$
 Berge strategy yields the best payoffs to the others' players who also play Berge strategies.  
If the players have chosen a strategy profile that forms a Berge solution, and $i$ sticks to the chosen strategy but some of the other players change their strategies, then $i$'s payoff will not increase. This is a resilience to  (single or joint) deviation by other players or other teams. 
 \begin{proposition}
 In the prisoner's dilemma game,  the strategy profile
$(1,1)$ is the unique Berge solution.
\end{proposition}
\begin{proof} By definition of prisoner's dilemma game, one has the following inequalities:
$a_{12}<a_{22}< a_{11}<a_{21},$ and
 $b_{21}<b_{22}< b_{11}<b_{12}.$
 In particular
$a_{11}>a_{12}$ and $  b_{11}> b_{21}.$  This means $(1,1)$ is a Berge solution. It is easy to check that there is no other Berge solution in the prisoner's dilemma game.
\end{proof}

This result is important because the outcome $(1,1)$ appears as a mutual support between the relay nodes. Note that such an outcome is not 
possible in the Nash  prisoner's dilemma game. The Nash equilibria is not able to predict observed outcomes in practice in the unmodified game
 while the 
Berge solution is predicting a better outcome  as $(1,1)$ is observed in many experimental setups even in the one-shot game case.
Berge solution occurs when players are mutually supportive in the prisoner's dilemma game.  It means that Player 1 supports Player 2 and Player 2 supports Player 1. The next proposition establishes a connection between positive partial mutual altruism and  Berge solution in $2\times 2$ games with dominant strategy

\begin{proposition}
Positive partial mutual altruism leads to the Berge solution in the prisoner's dilemma game.
\end{proposition}
\begin{proof}
Let $\lambda_{12}$ and $\lambda_{21}$  be both positive and of high level. The Nash outcomes of the empathetic prisoner's dilemma game are summarized in Table \ref{tab:sum2}.

\begin{table}[htb]
\begin{tabular}{|c|c|c|c|c|}
  \hline
  %% after \\: \hline or \cline{col1-col2} \cline{col3-col4} ...
 Player 1 $\backslash$ Player 2&  $\lambda_{21}$  Negative & Low & Medium &  High \\ \hline
 $\lambda_{12}$  {\color{green} High} & { \color{green} 1}2 &  {\color{green} 1}2 & {\color{green} 1 1}&   {\color{green} 11}\\
 $\lambda_{12}$  Medium & 22 & 22& {\color{green} 11},22, x{ \color{green}1}+(1-x)2 & {\color{green} 11} \\
  $\lambda_{12}$  Low & 22 & 22& 22&  2{\color{green} 1}\\ \hline
   $\lambda_{12}$ Negative & 22&  22 &  22&2{\color{green} 1}\\   \hline
\end{tabular}
 \caption{Summary of the outcomes. When $\lambda_{12}$ and $\lambda_{21}$  are both positive and of high level, the Nash equilibria of the empathetic game coincides with the  Berge solution  $(1,1).$   }\label{tab:sum2} 
\end{table}
It follows that when $\lambda_{12}$ and $\lambda_{21}$  are both positive and of high level, the Nash equilibria of the empathetic game coincides with the  Berge solution  $(1,1).$ This completes the proof.
\end{proof}

%%\section*{Discussion and future direction}

\section{Inconsistency of empathy structure} \label{incon} 

In this section we examine the consistency of some empathy profiles at different level of reasoning. The 1-level game is the one obtained by applying the matrix $\Lambda$ to the payoff vector $ r=\left(\begin{array}{c} r_1\\  r_2\end{array}\right)=r^{I}.$ Thus,
$$
r^{\Lambda,1}=r^{\Lambda}=(r^{\Lambda}_1,r^{\Lambda}_2)=\Lambda r,\ r^{\Lambda,0}:=r,
$$
At the $k-$th level of empathy the game payoff vector becomes
$$
r^{\Lambda,k}=\Lambda r^{\Lambda,k-1}=\Lambda^k r.
$$

\begin{definition}
The empathy structure is consistent if the equilibrium structure of $k-$th level game is unchanged for any $k\geq 1.$
\end{definition}

\begin{example}[Consistent empathy profile]
The identity matrix  $\left(\begin{array}{cc} 1 & 0\\  0 & 1\end{array}\right)$ is a consistent empathy structure.
\end{example}

\begin{example}[Consistent empathy profile]
Let $\rho>0,$ and $\Lambda=\frac{\rho}{2}\left(\begin{array}{cc} 1 & 1\\  1 & 1\end{array}\right).$ Then for any $k\geq 1,$  $ \Lambda^k=\rho^{k-1} \Lambda.$
It turns out  the payoff vector at the $k-$th level empathetic game is 
$$
r^{\Lambda,k}=\Lambda^k r=\rho^{k-1} \Lambda r.
$$
Since $\rho^{k-1}>0$  the $k-$level empathy game is strategically equivalent to the $1-$level empathy game. This means that $\frac{\rho}{2}\left(\begin{array}{cc} 1 & 1\\  1 & 1\end{array}\right)$ is consistent for $\rho>0,$ i.e. in the partially equally altruism case.
\end{example}

\begin{example}[Inconsistent empathy profile]
Let $\rho<0,$ and $\Lambda=\frac{\rho}{2}\left(\begin{array}{cc} 1 & 1\\  1 & 1\end{array}\right).$ Then for any $k\geq 1,$ $ \Lambda^k=\rho^{k-1} \Lambda.$
It turns out  the payoff vector at the $(2k+1)$-th level empathetic game is 
$$
r^{\Lambda,2k+1}=\Lambda^{2k+1} r=\rho^{2k} \Lambda r.
$$
Since $\rho^{2k}>0$  the $(2k+1)$-level empathy game is strategically equivalent to the $1-$level empathy game. This means that $\frac{\rho}{2}\left(\begin{array}{cc} 1 & 1\\  1 & 1\end{array}\right)$ 

However $\rho^{2k+1}<0$  the $2k$-level empathy game may not be strategically equivalent to the $1-$level empathy game. This means that $\frac{\rho}{2}\left(\begin{array}{cc} 1 & 1\\  1 & 1\end{array}\right)$ 
is INconsistent for $\rho<0$  and $r=r^{I}$ non-trivial vectorial function.
\end{example}

\begin{proposition}[Sufficient condition for consistence]
If there exists a positive sequence $\epsilon_k>0$ such that $\Lambda^k=\epsilon_k \Lambda$ then  the empathy structure $\Lambda$ is consistent.
\end{proposition}

\begin{proof}
The proof is immediate. Let  $\epsilon_k>0$ and the matrix $\Lambda$ satisfying the relation  $\Lambda^k=\epsilon_k \Lambda.$  This means that the payoff vector of the $k$-th level game  is
 $r^{\Lambda,k}= \epsilon_k \Lambda r.$ Since $\epsilon_k>0$ the $k-$th level empathy game is strategically equivalent to the $1-$level empathy game for any $k\geq 1.$
\end{proof}

\begin{proposition} \label{prott} The empathy structures  $$ \left(\begin{array}{cc} \lambda_{11} & \lambda_{12} \\  \lambda_{12}  & \lambda_{22} \end{array}\right)$$
such that 
$$
\left\{\begin{array}{c} 
\lambda_{ii} \ \mbox{root of }\ \ x^2-\epsilon x+y=0,\\
\lambda_{12}\lambda_{21}=y,
\end{array}
\right.
$$
are solutions of the
 system $\Lambda^2=\epsilon \Lambda,\  \epsilon>0.$ These solutions are consistent empathy profiles. 
\end{proposition}
\begin{proof}
$$
\left(\begin{array}{cc} \lambda_{11} & \lambda_{12} \\  \lambda_{21}  & \lambda_{22} \end{array}\right).
\left(\begin{array}{cc} \lambda_{11} & \lambda_{12} \\  \lambda_{21}  & \lambda_{22} \end{array}\right)
=
\left(\begin{array}{cc} \lambda_{11}^2+ \lambda_{12}\lambda_{21}& \lambda_{12}(\lambda_{11}+\lambda_{22}) \\    \lambda_{21}(\lambda_{11}+\lambda_{22}) & \lambda_{22} ^2+\lambda_{12}\lambda_{21}\end{array}\right)$$

The matrix equation $\Lambda^2=\epsilon \Lambda$ becomes
$$
\left\{\begin{array}{c} 
 \lambda_{11}^2+ \lambda_{12}\lambda_{21}=\epsilon  \lambda_{11},\\
 \lambda_{12}(\lambda_{11}+\lambda_{22})=\epsilon  \lambda_{12},\\
  \lambda_{21}(\lambda_{11}+\lambda_{22})=\epsilon  \lambda_{21},\\
  \lambda_{22} ^2+\lambda_{12}\lambda_{21}=\epsilon \lambda_{22}.
\end{array}
\right.
$$ By choosing $\lambda_{12}\lambda_{21}=y$ one obtains
$$
\left\{\begin{array}{c} 
\lambda_{12}\lambda_{21}=y,\\
 \lambda_{11}^2+ y=\epsilon  \lambda_{11},\\
 \lambda_{12}(\lambda_{11}+\lambda_{22})=\epsilon  \lambda_{12},\\
  \lambda_{21}(\lambda_{11}+\lambda_{22})=\epsilon  \lambda_{21},\\
  \lambda_{22} ^2+y=\epsilon \lambda_{22}.
\end{array}
\right.
$$ 

which completes the proof.
\end{proof}

Notice  that for $x=\frac{\epsilon}{2}=\lambda_{ij},\ \ y=x^2,$  the empathy matrix  $\Lambda=\frac{\epsilon}{2}\left(\begin{array}{cc} 1 & 1\\  1 & 1\end{array}\right)$ is a solution to the  system of Proposition \ref{prott}. Similarly,
 $x=\epsilon=\lambda_{ii},\ \ y=0$ and  $\Lambda=\epsilon\left(\begin{array}{cc} 1 & 0\\  0& 1\end{array}\right)=\epsilon I,$ is a solution to the  system of Proposition \ref{prott}.

\subsection{Infinite hierarchy}
We now examine the limit of the matrix $\Lambda^k$ as $k$ grows without bound. As we have seen above $\Lambda^k$ may not converge in general. If the maximum modulus of the eigenvalues satisfy $\rho(\Lambda)<1,$ then $ \Lambda^k$ goes to $\left(\begin{array}{cc} 0 & 0\\  0 & 0\end{array}\right)$ as $k$ goes to infinity. Thus, the infinite hierarchy empathy game becomes a trivial one in this case.  If $ \lim_{k} \Lambda^k= \Lambda_{\infty}$ and the infinite hierarchy of empathy game payoff vector is $\Lambda_{\infty}r.$ The equilibrium structure of the game $ \Lambda_{\infty}r$ may be different than the equilibrium structure of the finite hierarchy of empathy with payoff vector  $\Lambda^k  r.$

\begin{proposition}[Infinitely Consistent Empathy Profiles]The only empathy profiles in a  generic $2\times 2$ empathetic bimatrix game which are infinitely consistent are $I$ and
and
 $ \left(\begin{array}{cc} \lambda_{11} & \frac{\lambda_{11}(1-\lambda_{11})}{\lambda_{21}}\\  \lambda_{21} & 1-\lambda_{11}\end{array}\right)$  for $(\lambda_{11},\lambda_{21})\in \mathbb{R}^2.$

\end{proposition}

\begin{proof} If the empathy matrix $\Lambda$ is diagonalizable but has a single eigenvalue, it must be $I.$ If $\Lambda$ is diagonalizable and has two eigenvalues, they must be 1 and 0. There is a basis of matrix $P$ such that $\Lambda=PDP^{-1}$ with $D=\left(\begin{array}{cc}1 & 0 \\ 0 & 
0\end{array}\right).$ So,  $$\Lambda^2=PD^2P^{-1}=PDP^{-1}=\Lambda$$
 Since the eigenvalues are  $$ \frac{tr(\Lambda)}{2}\pm \sqrt{\frac{tr(\Lambda)^2}{4} -\mbox{det} \Lambda},$$ it
implies that $\mbox{tr} \Lambda=1$ and $\mbox{det} \Lambda= 0.$ So given $\lambda_{11}$ and $\lambda_{21}$, row 1 must be a
multiple of row 2 by $\frac{\lambda_{11}}{ \lambda_{21}}.$
\end{proof}

\section*{Acknowledgment} The authors are grateful to the NYU  editing team for their comments that helped improve the initial version of this paper. 
This research work is supported by U.S. Air Force Office of Scientific Research under grant number FA9550-17-1-0259.
\section{Conclusion} \label{sec:five} 
We have presented novel methods that incorporate users' empathy in  $2\times 2$ matrix games. We have examined both empathy and antipathy, selfishness and selflessness in one single unified framework. We observed that empathy plays a crucial role in these games. It supports cooperation if the empathies signs are aligned with the payoffs signs. It helps in breaking limit cycling by adopting different empathy structure. It allows the survival of dominates strategies. It is shown that a wide range of empathetic evolutionary game dynamics converges to the set of ESS in non-trivial $2\times 2$ games. However, number of issues remain to be solved. Beyond these  promising preliminary results, we aim to examine outcomes and limitations for larger class of games such as mean-field-type games.
We have seen that more fairness and less inequity outcomes are possible thanks to the presence of partial altruism, empathy-cooperation and mutual support. Thus, the fairness can evolve, if for some reason a proportion of the population employs empathy. The method above has a disadvantage. It does not explain the evolution of empathy such as perspective taking, empathic concern,  fantasy scale, personal distress and involvement scale. Empirical evidence suggests that these scales are not complete and the empathy state of a player can evolve over time.  

%%\bibliography{myRef5}
\bibliographystyle{IEEEtran}

\end{document}